\documentclass[13pt,a4paper,fleqn]{article}
\usepackage[T1]{fontenc}
\usepackage{jheppub}
\usepackage{alemfonts}
\usepackage[utf8]{inputenc}

\usepackage{amsmath}
\usepackage{amsfonts}
\usepackage{amssymb}
\usepackage{xcolor}
\usepackage{mathrsfs}
\usepackage{amsthm}
\usepackage{float}
\usepackage{graphicx}
\usepackage[linesnumbered,lined,boxed,commentsnumbered]{algorithm2e}
\usepackage[caption = false]{subfig}
\usepackage[parfill]{parskip}
\usepackage[british]{babel}

\newcommand\numberthis{\addtocounter{equation}{1}\tag{\theequation}}

\newcommand{\avg}[1]{\ensuremath{\left\langle #1\right\rangle}}

\newcommand{\psq}[1]{\left[ #1\right]}
\DeclareMathOperator{\EX}{\mathbb{E}}% expected value
\numberwithin{equation}{section}

\long\def \beq#1\eeq {\begin{equation} #1 \end{equation}}
\long\def \beaq#1\eeaq {\begin{equation}\begin{aligned} #1 \end{aligned}\end{equation}}
\long\def \bes#1\ees {\begin{equation}\begin{split} #1 \end{split} \end{equation}}
\long\def \bea#1\eea {\begin{eqnarray} #1 \end{eqnarray}}
\long\def \bse[#1]#2\ese {\begin{subequations}\label{#1}\begin{align} #2 \end{align}\end{subequations}}

\newcommand{\mv}[1]{\langle #1\rangle}

\long\def\dm[#1]{\!\operatorname{d\mu}\left(#1\right)}

\theoremstyle{plain}
\newtheorem{Remark}{Remark}

\newtheorem{Proposition}{Proposition}

\newtheorem{Definition}{Definition}

\graphicspath{{./Figures/}}

\title{Interpolating between boolean and extremely high noisy patterns through Minimal Dense Associative Memories}

\author[a,b]{Francesco Alemanno}
\author[a]{Martino Centonze}
\author[a]{Alberto Fachechi}

\affiliation[a]{Dipartimento di Matematica e Fisica Ennio De Giorgi, Universit\`a del Salento, Lecce, Italy}
\affiliation[b]{CNR-Nanotec, Sezione di Lecce, Italy}
\affiliation[c]{Istituto Nazionale di Fisica Nucleare, Sezione di Lecce, Italy}
\abstract{Recently, Hopfield and Krotov introduced the concept of {\em dense associative memories} [DAM] (close to spin-glasses with $P$-wise interactions in a disordered statistical mechanical jargon): they proved a number of remarkable features these networks share and suggested their use to (partially) explain the success of the new generation of Artificial Intelligence.  Thanks to a remarkable ante-litteram analysis by Baldi \& Venkatesh, among these properties, it is known these networks can handle a maximal amount of stored patterns $K$ scaling as $K \sim N^{P-1}$.
\newline
In this paper, once introduced a {\em minimal dense associative network} as one of the most elementary cost-functions falling in this class of DAM, we sacrifice this high-load regime -namely we force the storage of {\em solely} a linear amount of patterns, i.e. $K = \alpha N$ (with $\alpha>0$)- to prove that, in this regime, these networks can correctly perform pattern recognition even if pattern signal is $O(1)$ and is embedded in a sea of noise $O(\sqrt{N})$, also in the large $N$ limit. To prove this statement, by extremizing the quenched free-energy of the model over its natural order-parameters (the various magnetizations and overlaps), we derived its phase diagram, at the replica symmetric level of description and in the thermodynamic limit: as a sideline, we stress that, to achieve this task, aiming at cross-fertilization among disciplines, we pave two hegemon routes in the statistical mechanics of spin glasses, namely the replica trick and the interpolation technique.
\newline
Both the approaches reach the same conclusion: there is a not-empty region, in the noise-$T$ vs load-$\alpha$ phase diagram plane, where these networks can actually work in this challenging regime; in particular we obtained  a quite high critical (linear) load in the (fast) noiseless case resulting in $\lim_{\beta \to \infty}\alpha_c(\beta)=0.65$.}
\begin{document}

\maketitle

\section{Introduction}
Due to an increase in the GPU processing power \cite{40,62}, availability of large data-sets for training stages \cite{WWW} and the (deep) multi-layer architectures where neural networks can finally be embedded \cite{Hinton1,arc1}, their impressive skills -overall termed {\em deep learning} \cite{DL0}- keep achieving successes in the most disparate fields of Science and Technology \cite{Ching} (particularly outperforming at work in biomedical imaging, where they -noawadays- detect patterns possibly earlier than humans \cite{NatRev}).
\newline
Despite a number of remarkable progresses (e.g. \cite{BarraPeter,Cocco,Decelle,Barbier,Huang,Florent,Remi,Metha,Mezard}), these computational successes yet lack a full theoretical bulk behind (e.g. as made available in the pairwise limit of shallow networks as Hopfield and Boltzmann machines \cite{Coolen}), hence the quest for a {\em rationale} where different know-how(s) possibly merge is nowadays mandatory in the agenda of several research groups, ranging from Computer Science to Applied Mathematics (possibly crossing Theoretical Physics at its proliferative intersection offered by Statistical Mechanics of Spin Glasses).
\newline
In these regards, recently, Hopfield and Krotov proposed as an underlying bridge between deep neural networks and dense associative memories \cite{HopfieldKrotov_DAM,HopfieldKrotov_DL} (the latter being  P-spin extensions \cite{P-spin1,P-spin2} of the celebrated Hopfield classical pairwise limit \cite{Hopfield}) proving how these higher-order cost functions are more robust against adversarial and rubbish inputs.
\newline
Furthermore, this class of neural networks was deeply analyzed by Venkatesh $\&$  Baldi and Bovier $\&$ Niederhauser in the past \cite{Baldi,BovierPspin} and it is known that -calling $K$ the number of patterns to handle, $N$ the amount of neurons to accomplish the task and $P$ the order of their interactions- their critical capacity scales as $K \propto N^{P-1}$ (and collapses to the standard one, i.e. $K = 0.14 N$, in the known pairwise limit of $P=2$ \cite{Amit}).
\newline
Recently some of the Authors addressed the Statistical Mechanical analysis of a generalized RBM introduced in the literature by Terrence Sejnowski in 1984 \cite{Semio}
%studied a modification of the standard Boltzmann machine, a {\em Sejnowski machine} \cite{Semio}  
and proved that it was able to perform pattern recognition of patterns whose intensity stays $O(1)$ even in a sea of noise $O(\sqrt{N})$ in the large $N$ limit  \cite{sej}. It was also shown a dual representation of this network in terms of a peculiar form of the class of models suggested by  Hopfield and Krotov \cite{sej}: as in the pairwise counterpart \cite{BarraEquivalenceRBMeAHN}, this duality played as a crucial step to explain this skill of these machines as they can be obtained by keeping the network's load away from the maximal regime (the Baldi $\&$ Venkatesh limit \cite{Baldi}). We stress that the inspection of the low-storage regimes for these networks already started in \cite{Albert1}.
\newline
Here we continue along this investigation \cite{Albert1,sej}, focusing on pattern recognition at extremely low signal-to-noise ratios, by  proving that such a skill is not peculiar to the Sejnowski machine (see Appendix C): still focusing on four-wise interactions among discrete neurons, it holds for a broader class of Hopfield and Krotov models (w.r.t. one used in \cite{sej}) that we call {\em Minimal Dense Associative Memory} (MDAM).  In particular we provide a phase diagram for the MDAM, at the replica symmetric level of description and in the linear-storage regime, to show that there is a huge region in the plane of the two tunable parameters -load $\alpha$ and noise $\beta$- where this phenomenon happens (here $\lim_{\beta \to \infty}\alpha_c(\beta)=0.65$).
\newline
For the sake of cross-fertilization, we present our results paving at first the standard route of the replica trick \cite{Coolen}, then confirming the picture obtained by the RS-ansatz by suitably adapting to the case a Guerra's interpolation scheme  \cite{Agliari-Barattolo,Guerra}.

\section{Minimal Dense Associative Memory}

Here we introduce a minimal cost function of the form suggested by Hopfield and Krotov, namely the {\em minimal dense associative memory} (MDAM).
\begin{Definition} %Once considered $K$ patterns $\bold{\xi}$ of length $N$ whose entries are boolean variables $\pm 1$, the model cost-function is defined as
\beq\label{Fra1}
\calH(\boldsymbol{\sigma}|\boldsymbol{\eta}) = -\frac{1}{2N^3}\sum_{\mu=1}^K\Big(\sum_{i,j=1}^N \eta^\mu_{ij} \sigma_i\sigma_j \Big)^2
\eeq
where $\sigma_i=\pm 1$, $i \in (1,...,N)$, are Ising spin and $\eta^\mu_{ij}$ is the symmetric synaptic tensor.
\end{Definition}
Our goal is to prove how this model can retrieve patterns of information also when they are immersed in a background of a $\mathcal{O}(\sqrt{N})$ gaussian noise. This result can be achieved by requiring the network to store only $\mathcal{O}(N)$ patterns instead of the theoretical upper limit of $\mathcal{O}(N^3)$. In order to do that, we introduce the following decomposition of the synaptic tensor
\begin{Definition}\label{etas} %The memory tensor of the model is defined as
The load of the network $\alpha$, as anticipated, is defined as
\beq
\alpha=\lim_{N\to \infty} \frac{K}{N},
\eeq
while, the signal+noise decomposition reads
\beq\label{decomposition}
\eta^\mu_{ij}=\frac{1}{\sqrt{1+\alpha}}(\xi^\mu_{ij}+\sqrt{K}J^\mu_{ij}),
\eeq
where $\xi^\mu_{ij}$ is the tensor, with entries $\pm 1$, constituting the "matrix" signal, while the noise is embedded in the symmetric tensor $J^\mu_{ij}$, whose entries are \emph{i.i.d.} $\mathcal{N}(0,1)$ variables.
%with the additional factorization constraint $\xi^1_{ij}=\xi^1_i\xi^1_j$.
\end{Definition}
\begin{Remark}
%The definition of $\alpha$ involves the thermodynamic limit; this is a standard definition for the load of a network. Along this paper, we fix this value to be constant in $N$.
The noise is given by the product $\sqrt{K}J^\mu_{ij}$, which in the thermodynamic limit globally scales as $\mathcal{O}(\sqrt{N})$, as $K\sim N$.
\end{Remark}
%\begin{Remark}
%One could be tempted to define a matrix signal $\xi_{ij}$. Anyway, this would force the system to align the product $\sigma_i \sigma_j$ in the direction given by $\xi_{ij}$, in order to maximize the "matrix magnetization"
%\beq
%M=\frac{1}{N^2} \sum_{i,j} \xi_{ij} \sigma_i \sigma_j \sim 1.
%\eeq
%However, this holds only for states that can be fully factorized as $\xi_{ij}=\xi_i \times \xi_j$. Roughly speaking, this is due to the fact that the matrix $\xi_{ij}$ has $\mathcal{O}(N^2)$ degrees of freedom, in contrast to the solely $\mathcal{O}(N)$ of a $N$-spin network. Therefore, it makes only sense to work directly with factorized patterns of information in the form $\xi_{ij}=\xi_i \times \xi_j$.
%\end{Remark}
We are interested in the study of the retrieval phase of the network; for simplicity we restrict ourselves to the study of the retrieval of pure states. The retrieved pattern is arbitrary: we just denote it by $\xi^1_{ij}$; the remaining states $\xi^\mu_{ij}$, with $\mu>1$, will then constitute a quenched noise for the system. Hence, we perform a quenched average over the $P-1$ remaining states $\xi^\mu_{ij}$ together with the amplified $J^\mu_{ij}$ noise, by introducing the following expectation operator
\begin{align}
\EX \equiv \Big(\prod_{i,j,\mu>1}^{N,N,K} \frac{1}{2} \sum_{\xi^\mu_{ij}=\pm 1}\Big) \Big(\prod_{i,j,\mu\geq 1}^{N,N,K} \int DJ^\mu_{ij}\Big).
\end{align}
As usual, all the thermodynamic properties can be derived from the quenched pressure\footnote{Notice that the pressure is strictly related to the quenched intensive free energy $f$ as $A=-\beta f.$}
\begin{Definition}
In the thermodynamic limit, the quenched pressure reads
\beq\label{pressure}
A_N=\lim_{N\to \infty}\frac{1}{N}\bbE \ln Z_N,%=\frac{1}{N}\bbE \ln \sum_{\sigma}\exp\psq{\frac{\beta}{2N^3}\sum_{\mu=1}^K\Big(\sum_{i,j=1}^N \eta^\mu_{ij} \sigma_i\sigma_j\Big)^2}.
\eeq
where $Z_N$ is the partition function, defined as
\begin{align}
Z_N=\sum_{\boldsymbol{\sigma}} \exp\left(-\beta \calH\right)=\sum_{\boldsymbol{\sigma}}\exp\Big(\frac{\beta}{2 N^3}\sum_{\mu=1}^{K}\big(\sum_{i,j=1}^{N} \eta^\mu_{ij}\sigma_i \sigma_j\big)^2\Big).
\end{align}
%Note that the partition function of the model is tacitely given by
%\beq
%Z=\sum_{\sigma}\exp\pto{\frac{\beta}{2N^3}\sum_{\mu=1}^K\Big(\sum_{i,j=1}^N \eta^\mu_{ij} \sigma_i\sigma_j \Big)^2}
%\eeq
%and the quenched averaging operation is defined as
\end{Definition}
%and where $\bbE_{\eta}$ is a boolean average for $\xi^1_{ij}$ and a gaussian average for $\xi^\mu_{ij}$ with $\mu\ge2$.
%\newline
\begin{Remark}
The partition function $Z_N$ can be written by introducing auxiliary Gaussian variables as follows
\begin{align}\label{Z}
Z_N=\int Dz \sum_{\boldsymbol{\sigma}}\exp\Big(\sqrt{\frac{\beta}{N^3}}\sum_{\mu=1}^{K} \sum_{i,j=1}^{N} \eta^\mu_{ij}\sigma_i \sigma_j z_\mu\Big),
\end{align}
where $\int Dz \equiv \int \prod_{\mu=1}^{K} Dz_\mu$ and $Dz_\mu$ the $\mathcal{N}(0,1)$ measure relative to the $\mu$ component of the vector $z_\mu$. We stress that, written in this form, the partition function is equivalent to that of a bi-partite system, with the hidden layer $\boldsymbol{z}$ added to the visible one, $\boldsymbol{\sigma}$. The hidden layer is therefore filled with real valued gaussian $\mathcal{N}(0,1)$ neurons.
\end{Remark}
In the following two sub-sections, we will tackle the problem of finding an explicit expression for the above pressure in the thermodynamic limit in terms of the natural order parameters of the model, defined only after having introduced $n$ replicas of the system (as usual in the context of replica trick and interpolation technique calculations).\par\medskip
\begin{Definition}
	The overlap $q_{ab}$ among two replicas ($a,b=1,..,n$) of the system is defined as
	\bes\label{q}
	q_{ab}&=\frac{1}{N}\sum_{i=1}^{N} \sigma^a_i\sigma_i^b.
	\ees
	Equivalently, the overlap relative to the hidden layer is defined as:
	\bes\label{p}
	p_{ab}&=\frac{1}{K-1}\sum_{\mu=2}^{K} z^a_\mu z_\mu^b.
	\ees
	The Mattis magnetization, for a generic pattern $\xi_i^\mu$, and for a generic replica $a$ of the system, reads:
	\bes\label{mattis}
	m^a_\mu&=\frac{1}{N}\sum_{i=1}^N \xi^\mu_{i} \sigma^a_i.
	\ees
	We here introduce the matrix magnetization (also relative to the $a$-th replica):
	\bes\label{matrix}
	M^a_\mu&=\frac{1}{N^2}\sum_{i,j=1}^N \xi^\mu_{ij} \sigma^a_i \sigma^a_j.
	\ees
\end{Definition}

%: extremizing the former over the latter, en-route for the self-consistencies of the order parameters,  we will obtain the phase diagram of the model, where we will see the existence of a not-empty region for pattern's retrieval \ref{Fra1}.

\subsection{Route One: Replica Trick}
The replica trick is based on the following identity:
\bes\label{trick}
A=\lim_{N\to \infty}\frac1N\bbE \ln Z_N=\lim_{N\to\infty}\lim_{n\to0} \frac{ \ln \bbE Z^n_N}{n N}.
\ees
%Introducing the following pattern decomposition:
%\begin{equation}
%\eta^\mu_{ij}=\frac{1}{\sqrt{1+\alpha}}\xi^\mu_{ij}+\sqrt{\frac{\alpha N}{1+\alpha}} J^\mu_{ij}
%\end{equation}
Introducing the decomposition (\ref{decomposition}), the $\EX Z_N^n$ partition function becomes
\bes
\EX Z^n_N=\Big(\prod_{a=1}^{n} \sum_{\boldsymbol{\sigma}^a}\Big) \int \Big(\prod_{a=1}^{n} D z^a\Big) \EX\exp\Big(&\sqrt{\frac{\beta}{(1+\alpha)N^3}}\sum_{a=1}^n\sum_{\mu=1}^{K} \sum_{i,j=1}^{N} \xi^\mu_{ij}\sigma_i \sigma_j z_\mu +\\ &+\sqrt{\frac{\beta \alpha}{(1+\alpha)N^2}}\sum_{a=1}^n\sum_{\mu=1}^{K} \sum_{i,j=1}^{N} J^\mu_{ij}\sigma_i \sigma_j z_\mu\Big).
\ees
We now assume that only a single pattern (say $\xi^1$) is candidate for retrieval. Therefore, all patterns with $\mu\ge 2$ will contribute to the noise. We can therefore factorize the signal ($\mu=1$) from the global noise ($\mu>1$) in the partition function. Thus, the quenched average of the $n$-th power of the partition function reads
%\begin{align}
%\EX Z^n=& \left(\prod_{a=1}^{n} \sum_{\boldsymbol{\sigma}^a}\right) \EX\int \left(\prod_{a=1}^{n} D z^a\right)\exp\left(\sqrt{\frac{\beta}{(1+\alpha)N^3}} \sum_{a=1}^n\sum_{i,j=1}^{N} \xi^1_{ij} \sigma^a_i \sigma^a_j z^a_1 +\sqrt{\frac{\beta \alpha}{(1+\alpha)N^2}}\sum_{a=1}^n\sum_{i,j=1}^{N} J^1_{ij}\sigma^a_i \sigma^a_j z^a_1\right)\times\\&\times \exp\left(\sqrt{\frac{\beta}{(1+\alpha)N^3}} \sum_{a=1}^n\sum_{\mu> 1}^{K} \sum_{i,j=1}^{N} \xi^\mu_{ij}\sigma^a_i \sigma^a_j z^a_\mu + \sqrt{\frac{\beta \alpha}{(1+\alpha)N^2}}\sum_{a=1}^n\sum_{\mu>1}^{K} \sum_{i,j=1}^{N} J^\mu_{ij}\sigma^a_i \sigma^a_j z^a_\mu \right)
%\end{align}
\bes
\EX Z^n=& \Big(\prod_{a=1}^{n} \sum_{\boldsymbol{\sigma}^a}\Big) \EX \exp\Big[\frac{\beta}{2(1+\alpha)}\sum_{a=1}^n \Big(\sqrt{N} M_1^a+ \frac{\sqrt{\alpha}}{N}\sum_{i,j=1}^{N} J^1_{ij}\sigma^a_i \sigma^a_j\Big)^2\Big]\times\\\times& \int \Big(\prod_{a=1}^{n} D \{z^a\}_{\mu>1}\Big)\exp\Big(\sqrt{\frac{\beta}{(1+\alpha)N^3}}\sum_{a=1}^n\sum_{\mu> 1}^{K} \sum_{i,j=1}^{N}\left( \xi^\mu_{ij} + \sqrt{\alpha N} J^\mu_{ij}\right)\sigma^a_i \sigma^a_j z^a_\mu\Big).
\label{eq:1}
\ees
In the first line, we can simply drop out the Gaussian contribution from the signal term, since
\begin{align}{\label{app1}}
\frac{1}{N}\sum_{i,j} J^1_{ij }\sigma^a_i \sigma^a_j \sim \mathcal{O}(1),
\end{align}
for each $a=1,\dots,n$.\footnote{Recall that only terms that are linear extensive in $N$ do contribute in the exponent, as lower order terms disappear in the thermodynamic limit. The $\sqrt{N} M_1^a$ term has the correct scaling
\begin{align}
\frac{1}{N^{3/2}}\sum_{i,j} \xi^1_{ij }\sigma^a_i \sigma^a_j \sim \mathcal{O}(N^{1/2}),
\end{align}
which becomes $\mathcal{O}(N)$, given the presence of the square in Eq. (\ref{eq:1}) and for this reason it cannot be neglected: it represents the signal in the network.}
Then, the we can split the $n$-th power of the partition function as
\begin{equation}\label{eq:split}
\EX Z^n=\Big(\prod_{a=1}^{n} \sum_{\boldsymbol{\sigma}^a}\Big) Z_{\text{signal}}Z_{\text{noise}},
\end{equation}
where
\begin{equation}
\begin{split}
Z_{\text{signal}}&=\exp\Big[\frac{\beta N}{2(1+\alpha)}\sum_{a=1}^n \left(M_1^a\right)^2\Big],\\
Z_{\text{noise}}&=\int \Big(\prod_{a=1}^{n} D \{z^a\}_{\mu>1}\Big)\EX\exp\Big(\sqrt{\frac{\beta}{(1+\alpha)N^3}}\sum_{a=1}^n\sum_{\mu> 1}^{K} \sum_{i,j=1}^{N}\left( \xi^\mu_{ij} + \sqrt{\alpha N} J^\mu_{ij}\right)\sigma^a_i \sigma^a_j z^a_\mu\Big).
\end{split}
\end{equation}

First, we focus on the signal term. The matrix magnetization \eqref{matrix} measures the overlap of the product $\sigma_i \sigma_j$ in the direction specified by the matrix $\xi_{ij}$.\footnote{We omit the "upper" index in $\xi^1_{ij}$, since it plays no role in what follows.} However, the network configuration is fixed by specifying the value of the $N$ variables $\sigma_i$, while the $\xi_{ij}$ has $\sim N^2$ degrees of freedom. This means that the network configuration could not retrieve a general tensor.\footnote{It can be shown that, when the pattern $\xi_{ij}$ is not fully factorized in the product of two copies of the same vector $\xi_i$, there are no possible spin configurations $\boldsymbol{\sigma}$ giving $|M|=1$. Roughly speaking, this is due to the fact that the matrix $\xi_{ij}$ has $\mathcal{O}(N^2)$ degrees of freedom, in contrast to the solely $\mathcal{O}(N)$ of a $N$-spin network.} This issue is removed by working directly with factorized information patterns, i.e. in the form $\xi_{ij}\equiv\xi_i \xi_j$. This has an interesting consequence: the matrix magnetization factorizes in the square of the Mattis magnetization:
\begin{align}
M^a=\Big(\frac{1}{N}\sum_{i=1}^N \xi_i \sigma_i^a\Big)^2 = (m_a)^2.
\end{align}
Hence, the signal term is simply\footnote{We neglect the irrelevant term $\left(\frac{N}{2\pi}\right)^{2n}$, as it gives no contribution in Eq. (\ref{trick}).}
\bes
Z_{\text{signal}}=\int\Big(\prod_a dm_a d\hat{m}_a\Big)\exp\Big(-iN\sum_a m_a \hat{m}_a-i\sum_a \hat{m}_a \sum_i \xi_i \sigma^a_i+N\frac{\beta}{2(1+\alpha)}\sum_a m^4_a\Big),
\ees
where $\hat{m}_a$ is the conjugated momentum of $m_a$, and naturally arises from the Fourier representation of the Dirac delta
\begin{equation}
1= \int \prod_a dm_a \delta (m_a -\tfrac1N \sum_i \xi_i \sigma_i ^a).
\end{equation}
Concerning the noise term, it can be evaluated as (see App. \ref{app:noise})
\begin{equation}
\begin{split}
Z_{\text{noise}}=&\int \Big(\prod_{a,b} dq_{ab} dp_{ab} d\hat{q}_{ab} d\hat{p}_{ab}\Big)\exp\Big( -\frac{\alpha N}{2} \ln \det ( \mathbb{I}+2i \hat{\mathbb{P}} )\Big)\times\\
\times&\exp\Big(i N \sum_{a,b}q_{ab}\hat{q}_{ab}-i\sum_{i} \sum_{a,b} \hat{q}_{ab} \sigma^a_i \sigma^b_i+i\alpha N\sum_{a,b} p_{ab} \hat{p}_{ab}+ \frac{\beta \alpha^2}{2(1+\alpha)} \sum_{a,b=1}^n q^2_{ab} p_{ab}\Big).
\end{split}
\end{equation}
Again, the parameters $\hat q_{ab}$ and $\hat p_{ab}$ are the conjugate momenta of $q_{ab}$ and $p_{ab}$.
Putting together our results, we end up with the following expression:
\bes\label{final1}
\EX Z^n=&\int \Big(\prod_{a,b} dq_{ab} dp_{ab} d\hat{q}_{ab} d\hat{p}_{ab}\Big)\Big(\prod_a dm_a d\hat{m}_a\Big)\exp\Big( -\frac{\alpha N}{2} \ln \det( \mathbb{I}+2i \hat{\mathbb{P}} )\\
&+i N \sum_{a,b}q_{ab}\hat{q}_{ab}+i\alpha N\sum_{a,b} p_{ab} \hat{p}_{ab}+\frac{\beta \alpha^2}{2(1+\alpha)} \sum_{a,b=1}^n q^2_{ab} p_{ab}\\
&+iN\sum_a m_a \hat{m}_a+N\frac{\beta}{2(1+\alpha)}\sum_a m^4_a\Big)\times\\
&\times\Big(\prod_{a=1}^{n} \sum_{\boldsymbol{\sigma}^a}\Big)\exp\Big(-i\sum_{i} \sum_{a,b} \hat{q}_{ab} \sigma^a_i \sigma^b_i-i\sum_a \hat{m}_a \sum_i \xi_i \sigma^a_i\Big).
\ees
The last line in the latter equation can be written as
\bes\label{final2}
&\Big(\prod_{a=1}^{n} \sum_{\boldsymbol{\sigma}^a}\Big)\exp\Big(-i\sum_{i} \sum_{a,b} \hat{q}_{ab} \sigma^a_i \sigma^b_i-i\sum_a \hat{m}_a \sum_i \xi_i \sigma^a_i\Big)=\\
%&=\exp\left(\sum_i \ln \left(\prod_{a=1}^{n} \sum_{\sigma^a_i=\pm 1}\right)\exp\left(-i \sum_{a,b} \hat{q}_{ab} \sigma^a_i \sigma^b_i-i\sum_a \hat{m}_a \xi_i \sigma^a_i\right)\right)=\\
&=\exp\Big[N\avg{\ln \Big(\prod_{a=1}^{n} \sum_{\sigma^a=\pm 1}\Big)\exp\Big(-i \sum_{a,b} \hat{q}_{ab} \sigma^a \sigma^b-i\sum_a \hat{m}_a \xi \sigma^a\Big)}_\xi\Big].
\ees
where the average over $\boldsymbol{\xi}$ has been defined as
\begin{equation}
\avg{g(\boldsymbol{\xi})}_{\xi} \equiv\lim_{N\to\infty} \frac{1}{N} \sum_{i=1}^{N} g(\boldsymbol{\xi}_i).
\label{eq:av}
\end{equation}
Assuming the commutativity of the two limits $N\to\infty$ and $n\to0$ (following the replica trick paradigm \cite{Mingione}), we can compute the statistical pressure in the thermodynamic limit through the saddle point method, which gives
\begin{align}\label{phin}
A=\lim_{n\to 0} \frac{1}{n} \text{Extr} \,\phi,
\end{align}
where $\phi$ is the argument of the exponential in the partition function (see Eqs. (\ref{final1}) and (\ref{final2})), i.e.:
\bes
\phi&=i \sum_{a,b}q_{ab}\hat{q}_{ab}+i\alpha \sum_{a,b} p_{ab} \hat{p}_{ab}+\frac{\beta \alpha^2}{2(1+\alpha)} \sum_{a,b=1}^n q^2_{ab} p_{ab}-\frac{\alpha}{2} \ln \det( \mathbb{I}-2i \hat{\mathbb{P}} )+\\
&+i\sum_a m_a \hat{m}_a+\frac{\beta}{2(1+\alpha)}\sum_a m^4_a+\avg{\ln \sum_{\boldsymbol{\sigma}}\exp\Big(-i\sum_{a,b} \hat{q}_{ab} \sigma^a \sigma^b-i\sum_a \hat{m}_a \xi \sigma^a\Big)}_\xi.
\ees
We can drop out the conjugates momenta by imposing the saddle point conditions on $p$,$q$ and $m$ respectively, which correspondingly give
\begin{align}
\hat{p}_{ab}=\frac{i}{2} \frac{\beta \alpha}{1+\alpha} q^2_{ab}, \qquad \hat{q}_{ab}=i \frac{\beta \alpha^2}{1+\alpha} q_{ab} p_{ab}, \qquad \hat{m}_a=2i\frac{\beta}{1+\alpha}m_a^3.
\end{align}
With these conditions, we obtain a simpler form for $\phi$, namely:
\bes\label{phi}
{\phi}=&-\frac{\beta \alpha^2}{1+\alpha} \sum_{a,b}q^2_{ab} p^2_{ab}-\frac{\alpha}{2} \ln \det( \mathbb{I}+\frac{\beta \alpha}{1+\alpha} \mathbb{Q})-\frac{3}{2}\frac{\beta}{1+\alpha} \sum_a m_a^4+\\
&+\avg{\ln \sum_{\boldsymbol{\sigma}}\exp\Big(\frac{\beta \alpha^2}{1+\alpha}\sum_{a,b} q_{ab} p_{ab} \sigma^a \sigma^b-\frac{2\beta}{1+\alpha} \xi \sum_a m^3_a \sigma^a\Big)}_\xi,
\ees
where the matrix $\mathbb{Q}$ has been defined as $\mathbb{Q}_{ab}\equiv q^2_{ab}$.\\
\begin{Definition}
The replica symmetric ansatz (\emph{RS}) for this network model reads
\begin{align}\label{RS}
q_{ab}=\delta_{ab}+q(1-\delta_{ab}), \qquad p_{ab}=p_D\delta_{ab}+p(1-\delta_{ab}), \qquad m_a=m.
\end{align}
\end{Definition}
We are now able to enunciate the following proposition regarding the quenched pressure (\ref{pressure}) in the {RS} ansatz:
\par\medskip
\begin{Proposition}
The replica symmetric expression of the quenched pressure related to the model (\ref{Fra1}) reads
\bes\label{A}
A^{RS}=&\ln 2-\frac{\beta \alpha^2}{1+\alpha} (qp -q^2 p) - \frac{\alpha}{2} \ln \Big(1-\frac{\beta \alpha}{1+\alpha} (1-q^2) \Big)+ \frac{\alpha}{2}\frac{\beta \alpha}{1+\alpha}\frac{q^2}{1-\frac{\beta \alpha}{1+\alpha} (1-q^2) }+\\
&-\frac{3}{2}\frac{\beta}{1+\alpha} m^4+\int Dx \:\ln \cosh \Big( \sqrt{2\frac{\beta \alpha^2}{1+\alpha} p q}\:x + \frac{2\beta}{1+\alpha} m^3 \Big).
\ees
\end{Proposition}
\begin{proof}
The details of the RS ansatz computations are reported in App. \ref{app:RS}.
\end{proof}

\subsection{Route Two: Interpolation method}
We now proceed to check the validity of the replica trick computation with an alternative route, i.e. the Guerra's interpolation method. Given the expression of $Z_N$ in Eq. \eqref{Z}, and substituting the explicit form of $\eta$ (according to Def. \ref{etas}) in terms of signal and noise, the statistical pressure in the thermodynamic limit reads
\bes
A=\lim_{N\to \infty}\frac{1}{N}\bbE \ln \sum_{\sigma}\int Dz\:\exp\Big(\sqrt{\frac{\beta}{(1+\alpha)N^{3}}}&\sum_{\mu=1}^K\sum_{i,j=1}^N \xi^\mu_{ij} \sigma_i\sigma_jz_\mu+\\
&+\sqrt{\frac{\beta \alpha}{(1+\alpha)N^{2}}}\sum_{\mu=1}^K\sum_{i,j=1}^N J^\mu_{ij} \sigma_i\sigma_jz_\mu \Big).
\ees
Again, we isolate the signal ($\mu=1$) from the noise ($\mu>1$), always neglecting the irrelevant term because of Eq. (\ref{app1}). Thus
%Now, in order to proceed with our calculations, we split the sum over $\mu$ to consider the different statistical nature of $\xi^1$ and $\xi^2\cdots \xi^K$, furthermore we will marginalize over $z_1$, obtaining
\bes
A=\lim_{N\to \infty}&\frac{1}{N}\bbE \ln \sum_{\sigma}\int D z\exp\Big(\sqrt{\frac{\beta}{(1+\alpha)N^{3}}}\sum_{\mu>1}^K\sum_{i,j=1}^N \xi^\mu_{i} \xi^\mu_{j} \sigma_i\sigma_jz_\mu+\\
&+\sqrt{\frac{\beta \alpha}{(1+\alpha)N^{2}}}\sum_{\mu>1}^K\sum_{i,j=1}^N J^\mu_{ij} \sigma_i\sigma_jz_\mu +\frac{\beta N}{2(1+\alpha)}\big(\frac{1}{N}\sum_{i,j=1}^N \xi^1_i  \sigma_i\big)^4\Big).
\ees
Notice that we already adopted the signal factorization $\xi_{ij}=\xi_i \xi_j$, which allows us to directly express everything in terms of the Mattis magnetization $m_1$ associated to the retrieved pattern $\mu=1$.
\par
%The squared quantity is very important as it carries the signal and it is useful to introduce the following
%\newline
%\textbf{Dovreste forse convenire dove definire tutto -idealmente prima delle sottosezioni con i due metodi, dove si danno le definizioni iniziali- e forse mettere  li anche quelle delle M e degli overlaps, vedete voi...}
%\begin{Definition} The natural magnetization order parameter for the model coded by the cost function (\ref{Fra1}) is
%\beq
%m_\mu(\sigma|\xi)\equiv \frac{1}{N^2}\sum_{i,j=1}^N \xi^\mu_{ij} \sigma_i\sigma_j
%\eeq
%\end{Definition}
We are now ready to set up the interpolation strategy. We introduce an interpolating parameter $t \in (0,1)$ such that (in its extrema) it compares the original model (recovered for $t=1$) and a {\em simpler} model at $t=0$. Hence we introduce the next\medskip
\begin{Definition}
The Guerra's interpolating pressure for the MDAM coded by the cost function \eqref{Fra1} reads as
\bes\label{eq:intA}
\calA(t)=\lim_{N\to \infty}&\frac{1}{N}\bbE \ln \sum_{\sigma}\int D z\exp\Big(\sqrt{t}\sqrt{\frac{\beta}{(1+\alpha)N^{3}}}\sum_{\mu=1}^K\sum_{i,j=1}^N \xi^\mu_{ij} \sigma_i\sigma_jz_\mu+\\
&+\sqrt{t}\sqrt{\frac{\beta \alpha}{(1+\alpha)N^{2}}}\sum_{\mu=1}^K\sum_{i,j=1}^N J ^\mu_{ij} \sigma_i\sigma_jz_\mu +t\frac{\beta N}{2(1+\alpha)}m_1^4+\sqrt{1-t} \calW +(1-t)\calD\Big),
\ees
where $\calW$ and $\calD$ are defined as
\begin{align}
\calW=&\sqrt{\frac{\beta}{1+\alpha}} C_1\sum_iJ_i\sigma_i+\sqrt{\frac{\beta}{1+\alpha}} C_2\sum_\mu J_\mu z_\mu,\\
\calD=&C_3\frac{\beta}{1+\alpha}\sum_\mu \frac{z^2_\mu}{2}+C_4 \frac{\beta N}{(1+\alpha)} m_1,
\end{align}
and $C_1,\ ...,\ C_4$ are constants whose explicit values will be set later, see Eq. \eqref{stamberga}.
\newline
The \emph{interpolating variables} $J_i$ and $J_\mu$ are, respectively, $N$-component and $K$-component vectors of {i.i.d.} $\mathcal{N}(0,1)$ variables. Therefore, the expectation $\EX$ is now extended to include these new degrees of freedom.
\end{Definition}
\begin{Proposition}
The quenched pressure related to the model \eqref{Fra1} can thus be recovered using the Fundamental Theorem of Calculus:
\beq
A=\calA(t=1)=\calA(t=0)+\int_0^1 dt\, \partial_t \calA(t).
\label{eq:sumrule}
\eeq
\end{Proposition}
Following the scheme used in \cite{Albert2,Guerra,Huang}, we can evaluate separately $\partial_t\calA$ and $\calA(0)$. Tackling the $t$-streaming first and keeping as order parameters those defined in equations (\ref{q}), (\ref{p}) and (\ref{matrix}),
%\bes
%q_{ab}&=\frac{1}{N}\sum_i \sigma^a_i\sigma_i^b,\\
%p_{ab}&=\frac{1}{K-1}\sum_i z^a_\mu z_\mu^b,\\
%m_\mu&= \frac{1}{N^2}\sum_{i,j=1}^N \xi^\mu_{ij} \sigma_i\sigma_j,
%\ees
we obtain
%\bes
%\partial_t\calA=\frac{\beta}{2(1+\alpha)}\bbE \big[& \alpha(\alpha+\frac{1}{N})(\mv{p_{11}}-\mv{p_{12}q_{12}^2}) -C_1^2 +C_1^2\mv{q_{12}} -\alpha C_2^2 \mv{p_{11}} +\alpha C_2^2 \mv{p_{12}} +\\
%& -\alpha C_3\mv{p_{11}} + \mv{m_1^2}- 2C_4\mv{m_1}\big].
%\ees
%We can already disregard the term proportional to $N^{-1}$ since it only plays a role at finite $N$, and we want to take the thermodynamic limit ($N\to\infty$), thus we must only deal with
\bes\label{LaDue}
\partial_t\calA(t)=\frac{\beta}{2(1+\alpha)}\bbE \big[& \alpha^2\mv{p_{11}}-\alpha^2\mv{p_{12}q_{12}^2} -C_1^2 +C_1^2\mv{q_{12}} -\alpha C_2^2 \mv{p_{11}} +\alpha C_2^2 \mv{p_{12}} +\\
& -\alpha C_3\mv{p_{11}} + \mv{m_1^4}- 2C_4\mv{m_1}\big],
\ees
where we use the standard notation $\avg{.}$ for the Boltzmann average.\footnote{Notice that the Boltzmann average has a functional dependence from the interpolating parameter $t$, as every thermodynamic observable is computed from the general interpolating pressure (\ref{eq:intA}).} The terms involving the Mattis magnetizations for $\mu>1$ do not appear in the streaming equation since their contribution is subleading in the thermodynamic limit.\\
The expected Boltzmann averages $\mv{q_{12}} $, $\mv{p_{12}}$, $\mv{p_{11}}$ and $\mv{m_1}$ are difficult to compute, but recall that we are interested in the replica symmetric evaluation of the quenched free energy (and, thus, also of the replica symmetric expression of all the order parameters). Introducing the fluctuations of the order parameter (centered around their quenched mean values $q$, $p$ and $m$, see eq. \ref{RS})
\begin{align}\label{fluct}
\Delta_q&=q_{12}-q,\\
\Delta_p&=p_{12}-p,\\
\Delta_m&=m_{1}-m,
\end{align}
and recalling that, in the $RS$ approximation, they vanish in the thermodynamic, we can rewrite the interaction terms in eq. (\ref{LaDue}) as
\bes
\mv{p_{12}q_{12}^2} &= -2pq^2 + q^2\mv{p_{12}} + 2pq\mv{q_{12}},\\% + 2q\mv{\Delta_p\Delta_q}+p\mv{\Delta_q^2}+\mv{\Delta_p\Delta_q^2}\\
\mv{m_1^4}&=- 3 m^4 + 4 m^3 \mv{m_1}.
\ees
%Allowing for a powerful simplification in the replica symmetric regime where fluctuations are vanishing, thus the interaction terms become
%\bes
%\mv{p_{12}q_{12}^2} &= -2pq^2 + q^2\mv{p_{12}} + 2pq\mv{q_{12}}\\
%\mv{m_1^4}&=-3 m^4+ 4 m^3\mv{m_1}
%\ees
By substitution inside the streaming equation we obtain
\bes
\partial_t\calA=\frac{\beta}{2(1+\alpha)}\bbE \big[& \alpha^2\mv{p_{11}}+2\alpha^2 pq^2 - \alpha^2 q^2\mv{p_{12}} - 2\alpha^2 pq\mv{q_{12}} -C_1^2 +C_1^2\mv{q_{12}} -\alpha C_2^2 \mv{p_{11}}  +\\
&+\alpha C_2^2 \mv{p_{12}} -\alpha C_3\mv{p_{11}} - 3 m^4 + 4 m^3 \mv{m_1}- 2C_4\mv{m_1}\big].
\ees
Recall that we have four free parameters: $C_1,\:C_2,\:C_3$ and $C_4$. They can be chosen \emph{a fortiori} in order to eliminate the expected Boltzmann averages $\mv{q_{12}} $, $\mv{p_{12}}$, $\mv{p_{11}}$ and $\mv{m_1}$ in favour of their replica-symmetric expectations in the thermodynamic limit.	With  this idea in mind, we rewrite the latter equation as
\bes
\partial_t\calA=\frac{\beta}{2(1+\alpha)}\bbE \big[& (\alpha^2  -\alpha C_2^2  -\alpha C_3 )\mv{p_{11}}  +(\alpha C_2^2- \alpha^2 q^2)\mv{p_{12}} +(C_1^2  - 2\alpha^2 pq)\mv{q_{12}} +\\
&+ (4 m^3 - 2C_4)\mv{m_1}-C_1^2  + 2\alpha^2 pq^2 - 3 m^4 \big].
\ees
It is now clear that, with the following choice:
\bes\label{stamberga}
C_1=\sqrt{2\alpha^2 p q}, \qquad C_2 =\sqrt{\alpha}q, \qquad C_3=\alpha(1-q^2), \qquad C_4=2 m^3,
\ees
we can achieve our goal and simplify the streaming term further, obtaining
\bes\label{stream}
\partial_t\calA=-\frac{\beta}{2(1+\alpha)}\big[2\alpha^2 pq(1-q) + 3 m^4 \big].
\ees
Now we are left with the one body term:
\bes
\calA(0)=\lim_{N\to \infty} \frac{1}{N}\bbE \ln \sum_{\sigma}\int D z\exp\Big( &\sqrt{\frac{\beta}{1+\alpha}} C_1\sum_iJ_i\sigma_i+\sqrt{\frac{\beta}{1+\alpha}} C_2\sum_\mu J_\mu z_\mu+\\
&+C_3\frac{\beta}{1+\alpha}\sum_\mu \frac{z^2_\mu}{2}+C_4 \frac{\beta }{1+\alpha} \sum_{i=1}^N \xi^1_{i} \sigma_i\Big).
\ees
It can be easily computed, returning
\bes\label{A0}
\calA(0)=&- \frac{\alpha}{2} \ln \Big(1-\frac{\beta \alpha}{1+\alpha} (1-q^2) \Big)+ \frac{\alpha}{2}\frac{\beta \alpha}{1+\alpha}\frac{q^2}{1-\frac{\beta \alpha}{1+\alpha} (1-q^2) }\\&+\int Dx \:\ln \cosh \Big( \sqrt{2\frac{\beta \alpha^2}{1+\alpha} p q}\:x + \frac{2\beta}{1+\alpha} m^3 \Big).
\ees
Combining equations (\ref{stream}) and (\ref{A0}), after some rearrangements we get the same result already derived through the replica trick for the {RS} pressure (\ref{A}).

\subsection{Phase Diagram}
Before moving on, we rewrite here the \emph{RS} pressure for the reader's convenience:
\bes
A^{RS}=&\ln 2-\frac{\beta \alpha^2}{1+\alpha} (qp -q^2 p) - \frac{\alpha}{2} \ln \Big(1-\frac{\beta \alpha}{1+\alpha} (1-q^2) \Big)+ \frac{\alpha}{2}\frac{\beta \alpha}{1+\alpha}\frac{q^2}{1-\frac{\beta \alpha}{1+\alpha} (1-q^2) }+\\
&-\frac{3}{2}\frac{\beta}{1+\alpha} m^4+\int Dx \:\ln \cosh \Big( \sqrt{2\frac{\beta \alpha^2}{1+\alpha} p q}\:x + \frac{2\beta}{1+\alpha} m^3 \Big).
\ees
Extremizing the statistical pressure with respect to the parameters $q,p$ and $m$, we end up with the self-consistency equations
%g^2&=\frac{2M^2 \beta}{1+\alpha}\\
\begin{equation}\label{eq:selfc}
\begin{split}
m  &=  \int Dx\,\tanh\Big(\sqrt{2\frac{ \beta \alpha^2}{1+\alpha} p q }\:x + \frac{ 2 \beta }{1+\alpha}\:  m^3\Big),\\
q  &=  \int Dx\,\tanh^2\Big(\sqrt{2\frac{ \beta \alpha^2}{1+\alpha} p q }\:x + \frac{ 2 \beta }{1+\alpha}\:  m^3\Big),\\
p&=\frac{\frac{\beta\alpha}{1+\alpha}  q^2  }{\Big(1- \frac{\beta\alpha}{1+\alpha} \left(1-q^2\right) \Big)^2}.
\end{split}
\end{equation}
We numerically solve these equations and paint the phase diagram reported in Fig.\ref{phdiag} (left panel), made by three different phases: the Retrieval (R), characterized by non-zero values of the two order parameters $m$ and $q$, i.e. $m\neq0$ and $q\neq0$; the Spin Glass phase (SG), where $m=0$, $q\neq0$, and the Ergodic phase (E), with $m=q=0$. In the Retrieval region pure states are always global minima for the free energy. Since mixture states (which are present as well as pure ones) are not global minima of the free energy, we refer to the R phase as a "pure retrieval" phase. We argue that this is due to the decomposition \eqref{decomposition}.\\
Furthermore we performed Monte Carlo simulations in order to check out our assumptions (e.g. the RS ansatz). In particular, we focused on the E-R critical line (see the phase diagram, fig. \ref{phdiag} (left)). We performed a finite size scaling analysis (see figure \ref{phdiag}, (right)), which led to the critical temperatures for different values of $\alpha$ depicted in the bottom right panel, figure \ref{phdiag}, (right): numerical outcomes are in excellent agreement with the theoretical predictions.
%\end{Theorem}
%\begin{figure}[H]
%\centering
%\includegraphics[scale=0.3]{phase2.pdf}\includegraphics[scale=0.385]{confronto.pdf}
%\label{phdiag}
%\caption{(left) The phase diagram of the model obtained by solving the self-consistencies, see Eqs. \eqref{eq:selfc}. We highlight three region: a pure ergodic behavior (E), a spin glass phase (SG) and a retrieval one (R). (right) The phase diagram of the model obtained by Monte Carlo simulations at different sizes $N$, as reported in the legend of the figure: the finite size scaling depicted suggests the simulations to approach the theoretical line as the system size diverges (as expected).}%Dotted line is the Spin Glass Critical line, solid line is Ferromagnetic critical line: \textbf{dentro la fase di retrieval la linea di spin glass andrebbe levata perche' e' molto misleading}
%\end{figure}
%\begin{figure}[H]
%	\centering
%	\includegraphics[scale=0.75]{confronto.pdf}
%	\label{MCMC}
%	\caption{}
%\end{figure}

\begin{figure}[h!]
	\centering
	\begin{minipage}{.5\textwidth}
		\centering
		\includegraphics[width=\textwidth]{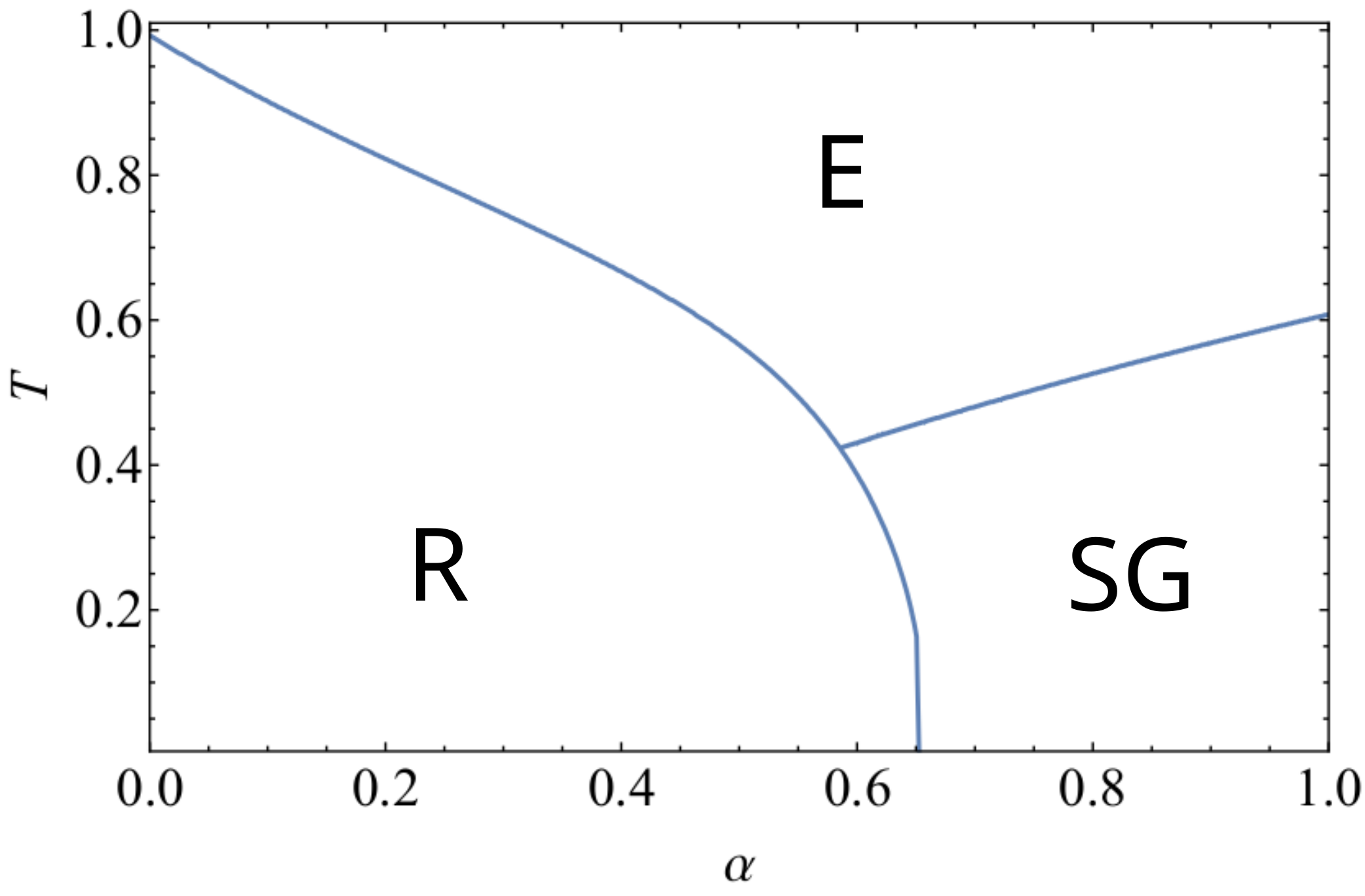}
			\vspace{0.05cm}
	\end{minipage}%
	\begin{minipage}{0.5\textwidth}
		\centering
		\includegraphics[width=\textwidth]{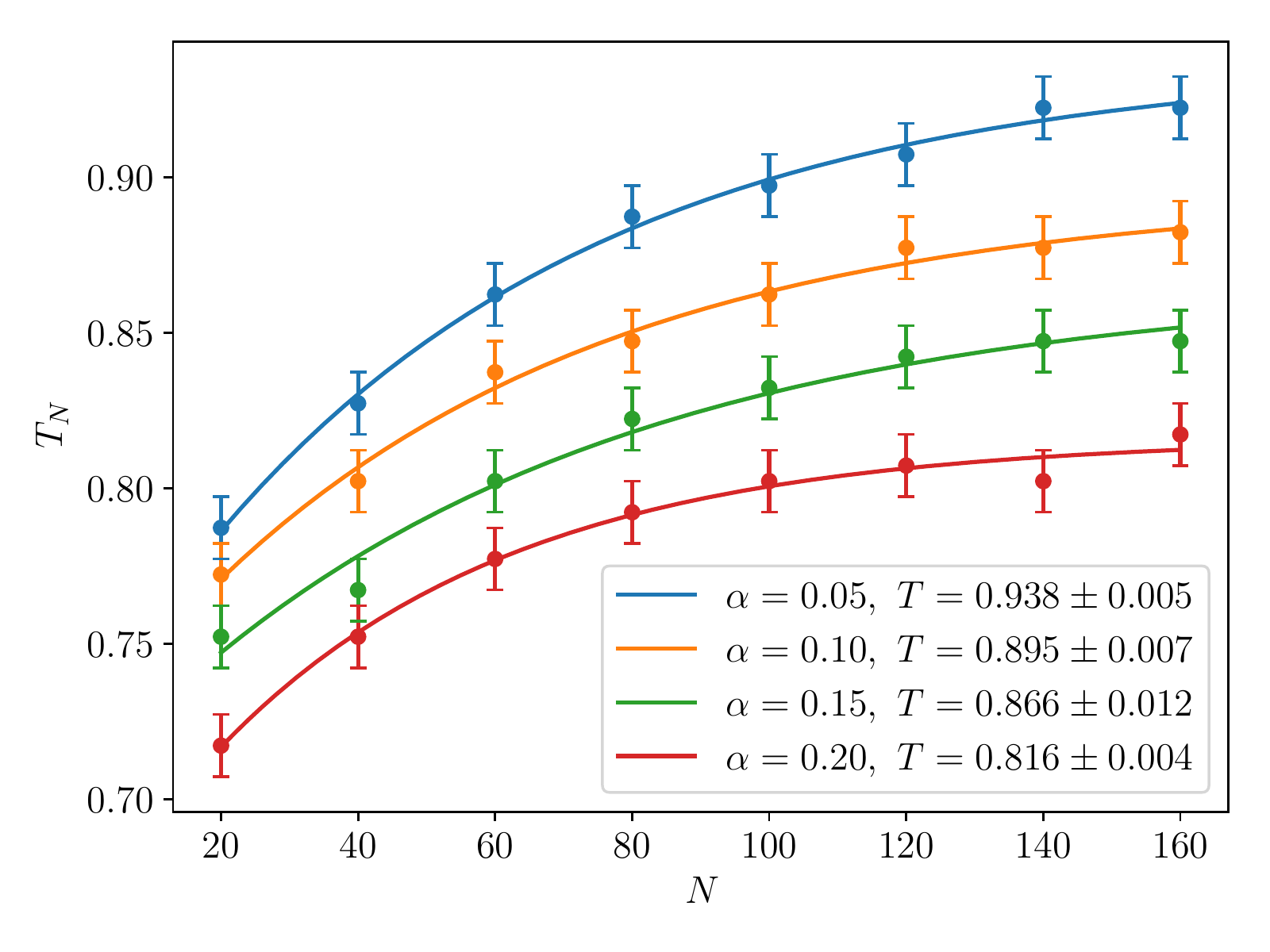}
		\vspace{0.0cm}
	\end{minipage}
\caption{(left) The phase diagram of the MDAM obtained by solving the self-consistencies, see Eqs. \eqref{eq:selfc}. We highlight three regions: a pure ergodic one (E), a spin glass phase (SG) and a retrieval one (R).  (right) Different critical lines $T_N$ depicting the E-R transition relative to different loads of the network, i.e. for $\alpha=0.05, 0.10, 0.15$ and $0.20$, as function of $N$. For each load $\alpha$ we performed Monte Carlo simulations for different sizes $N$, ranging from $N=20$ to $N=160$, with leaps of $\Delta N=20$, thus obtaining the different points interpolated in the right panel of the figure. The critical temperatures $T$ (for different values of $\alpha$) in the bottom right legend of the figure are so obtained; they are consistent with our theoretical results. 
	%The phase diagram of the model obtained by Monte Carlo simulations at different sizes $N$, as reported in the legend of the figure: the finite size scaling depicted suggests the simulations to approach the theoretical line as the system size diverges (as expected).
}
\label{phdiag}
\end{figure}

\section{Conclusions}

Along the lines of our recent research \cite{sej,Albert1}, in this paper we extensively relied upon tools typical of the statistical mechanics of spin-glasses to quantify the high pattern recognition capacity of, possibly, the simplest neural network falling in the class of {\em dense associative memories}. The latter were recently proposed by Hopfield and Krotov \cite{HopfieldKrotov_DAM}\cite{HopfieldKrotov_DL}  as a candidate benchmark to inspect for possibly explaining (part of) the impressive skills that artificial neural architectures experience nowadays.
\newline
In particular we have shown that such a network, equipped with {\em solely} a linear storage of patterns $K$ -in the volume $N$-  but where patterns are split in a $\mathcal{O}(1)$ signal term and an $\mathcal{O}(\sqrt{N})$ noisy term, is able to extensively de-noise the perceived inputs such as to accomplish pattern recognition despite the prohibitive level of noise: this is ultimately due to the dense connections where redundant representations of patterns are possible \cite{sej}. The critical capacity in this regime of operation -at least at the replica symmetric level of description- is quite huge, resulting in $\alpha_c (\beta \to \infty) \sim 0.65$ (and we checked numerically that the replica symmetric assumption is tolerated as shown by extensive Monte Carlo runs). In particular, at present -to our knowledge- this is the  largest critical capacity for networks presenting this high pattern recognition skill, (the network of \cite{sej} has to respect $\alpha_c  \leq 0.5$).
\newline 
Furthermore the retrieval region is characterized by the fact that pure states are always global minima for the free energy. We conjecture this is a consequence of the choice of the decomposition \eqref{decomposition}.
\newline
Finally, with the aim of promoting cross-fertilization among the two disciplines of Machine Learning and Disordered Statistical Mechanics, we collected the outlined results by using two among the most used methods to deal with spin-glasses, namely the replica trick \cite{Coolen} and the interpolation method \cite{Guerra}, discussing both of them in great detail.

\appendix

\section{Replica trick computations: details}
In this Appendix, we report some details on the replica trick computation.
\subsection{Evaluation of the noise term}\label{app:noise}
In this Section, we evaluate the noise term in the splitted partition function \eqref{eq:split}, which we report here for convenience:
\begin{equation}
Z_{\text{noise}}=\int \Big(\prod_{a=1}^{n} D \{z^a\}_{\mu>1}\Big)\EX\exp\Big(\sqrt{\frac{\beta}{(1+\alpha)N^3}}\sum_{a=1}^n\sum_{\mu> 1}^{K} \sum_{i,j=1}^{N}\left( \xi^\mu_{ij} + \sqrt{\alpha N} J^\mu_{ij}\right)\sigma^a_i \sigma^a_j z^a_\mu\Big).
\end{equation}
Because of the independence of the signal and noise term in the pattern decomposition \eqref{decomposition}, we can perform the averages separately. We start with performing the average over the $\xi$'s, which leads to
\begin{align}
\exp\Big(\sum_{i,j,\mu>1} \ln \cosh\big(\sqrt{\frac{\beta}{(1+\alpha)N^3}} \sum_{a} \sigma^a_i \sigma^a_j z^a_\mu\big) \Big).
\end{align}
In the large $N$ limit, we can expand in powers of $1/N$ the $\ln\cosh$ function, keeping only the leading contribution (as all higher order corrections vanish in the thermodynamic limit). Then, we are left with
\begin{align}
\exp\Big(\frac{\beta}{2(1+\alpha)N^3}\sum_{i,j,\mu>1}  \big(\sum_{a} \sigma^a_i \sigma^a_j z^a_\mu\big)^2 \Big).
\end{align}
However, the exponent in the latter equation is of order $\mathcal{O}(1)$, thus it is a subleading contribution w.r.t. to the Gaussian part of the noise term. Therefore, it can been neglected in the large $N$ limit. The result is
%We end up with
%\begin{align}
%\EX Z^n \simeq \left(\prod_{a=1}^{n} \sum_{\boldsymbol{\sigma}^a}\right)  Z_{signal} Z_{noise},
%\end{align}
%where
%\begin{align}
%Z_{signal}=\exp\left[\frac{\beta N}{2(1+\alpha)}\sum_{a=1}^n \left(M_1^a\right)^2\right]
%\end{align}
%and
\begin{align}
Z_{\text{noise}}=\int DJ \int \Big(\prod_{a=1}^{n} D \{z^a\}_{\mu>1}\Big)\exp\Big(\sqrt{\frac{\beta \alpha}{(1+\alpha)N^2}}\sum_{a=1}^n\sum_{\mu>1}^{K} \sum_{i,j=1}^{N} J^\mu_{ij}\sigma^a_i \sigma^a_j z^a_\mu \Big).
\end{align}
Now, we can directly average over the $J$ variables, obtaining
\bes
Z_{noise}&=\int \Big(\prod_{a=1}^{n} D \{z^a\}_{\mu>1}\Big)\exp\Big(\frac{\beta \alpha}{2(1+\alpha)N^2}\sum_{\mu>1}^{K} \sum_{i,j=1}^{N}\big(\sum_{a=1}^n \sigma^a_i \sigma^a_j z^a_\mu \big)^2\Big)=\\
&=\int \Big(\prod_{a=1}^{n} D \{z^a\}_{\mu>1}\Big)\exp\Big(\frac{\beta \alpha}{2(1+\alpha)N^2}\sum_{\mu>1}^{K} \sum_{i,j=1}^{N}\sum_{a,b=1}^n \sigma^a_i \sigma^b_i \sigma^a_j \sigma^b_j z^a_\mu z^b_\mu\Big).
\ees
In the last line, the dependence on the order parameters $q_{ab}$ and $p_{ab}$ is clear. Hence, we can now introduce a product of delta functions by using
\begin{align}
1=\int \Big(\prod_{a,b} dq_{ab} dp_{ab} \:\delta\big( q_{ab} - \frac{1}{N} \sum_{i=1}^N \sigma^a_i \sigma^b_i \big)\: \delta\big( p_{ab} - \frac{1}{K-1} \sum_{\mu=2}^K z^a_\mu z^b_\mu \big)\Big).
\end{align}
After this manipulation, we get
\begin{align*}
Z_{\text{noise}}=& \int \Big(\prod_{a=1}^{n} D \{z^a\}_{\mu>1}\Big)\Big(\prod_{a,b} dq_{ab} dp_{ab} \:\delta\big( q_{ab} - \frac{1}{N} \sum_{i=1}^N \sigma^a_i \sigma^b_i \big) \:\delta\big( p_{ab} - \frac{1}{K-1} \sum_{\mu=2}^K z^a_\mu z^b_\mu \big)\Big)\times\\
\times&\exp\Big(N \frac{\beta \alpha^2}{2(1+\alpha)} \sum_{a,b=1}^n q^2_{ab} p_{ab}\Big).
\numberthis
\end{align*}
At this point, we use the Fourier representation of the Dirac delta:
\begin{align}
\delta\big(q_{ab} - \frac{1}{N}\sum_i \sigma^a_i \sigma^b_i\big)=\frac{N}{2\pi}\int d \hat{q}_{ab} \exp\Big( i N \hat{q}_{ab} \big(q_{ab} - \frac{1}{N} \sum_i \sigma^a_i \sigma^b_i \big)\Big),
\label{eq:fourier}
\end{align}
and similarly for $p_{ab}$. Then, the noise term now reads:\footnote{Again, we neglect the irrelevant factors $\left(\frac{N}{2\pi}\right)^{n^2}\left(\frac{\alpha N}{2\pi}\right)^{n^2}$, as they give no contribution in Eq. (\ref{trick}).}
\bes
&\int \Big(\prod_{a,b} dq_{ab} dp_{ab} d\hat{q}_{ab} d\hat{p}_{ab}\Big)\Big(\prod_{a} D \{z^a\}_{\mu>1}\Big)\exp\Big(-i\sum_{\mu>1}\sum_{a,b} \hat{p}_{ab} z_\mu^a z_\mu^b\\
&+i N \sum_{a,b}q_{ab}\hat{q}_{ab}-i\sum_{i} \sum_{a,b} \hat{q}_{ab} \sigma^a_i \sigma^b_i+i\alpha N\sum_{a,b} p_{ab} \hat{p}_{ab}+\frac{\beta \alpha^2}{2(1+\alpha)} \sum_{a,b=1}^n q^2_{ab} p_{ab}\Big).
\ees
The integral over the $z$ variables can be easily performed, leading to
\begin{align}
\int  \Big(\prod_{a} D \{z^a\}_{\mu>1}\Big) \exp\Big( -i\sum_{\mu>1}\sum_{a,b} \hat{p}_{ab} z_\mu^a z_\mu^b \Big)=\prod_{\mu>1}^K \det( \mathbb{I}+2i \hat{\mathbb{P}} )^{-1/2},
\end{align}
where $\mathbb{I}_{ab}=\delta_{ab}$ is the $n\times n$ identity matrix and $\hat{\mathbb{P}}_{ab}\equiv \hat{p}_{ab}$. We therefore end with the final expression for the noise term
\begin{equation}
\begin{split}
Z_{\text{noise}}=&\int \Big(\prod_{a,b} dq_{ab} dp_{ab} d\hat{q}_{ab} d\hat{p}_{ab}\Big)\exp\Big( -\frac{\alpha N}{2} \ln \det ( \mathbb{I}+2i \hat{\mathbb{P}} )\Big)\times\\
\times&\exp\Big(i N \sum_{a,b}q_{ab}\hat{q}_{ab}-i\sum_{i} \sum_{a,b} \hat{q}_{ab} \sigma^a_i \sigma^b_i+i\alpha N\sum_{a,b} p_{ab} \hat{p}_{ab}+ \frac{\beta \alpha^2}{2(1+\alpha)} \sum_{a,b=1}^n q^2_{ab} p_{ab}\Big).
\end{split}
\end{equation}
\subsection{The replica symmetric ansatz}\label{app:RS}
In this Section, we compute term by term the contributions appearning in \eqref{phi} after adopting the RS ansatz. Since the statistical pressure presents an overall factor $1/n$ in (\ref{phin}), only the $\mathcal{O}(n)$ terms are relevant for our purposes (since we have to evaluate the $n\to0$ limit). For the first term, the leading contribution is
\begin{align}
\sum_{a,b}q^2_{ab} p^2_{ab}\sim n(p_D-p q^2).
\end{align}
For the second one, we have
\begin{equation}
\begin{split}
\ln \det\left( \mathbb{I}- \frac{\beta \alpha}{1+\alpha} \mathbb{Q} \right)\sim n \ln \left(1-\frac{\beta \alpha}{1+\alpha} (1-q^2) \right)- n\frac{\beta \alpha}{1+\alpha}\frac{q^2}{1-\frac{\beta \alpha}{1+\alpha} (1-q^2) }.
\end{split}
\end{equation}
The $m$-dependent contribution is trivial, and reads
\begin{align}
\sum_a m_a^4=n\: m^4.
\end{align}
Finally, the last term can be straightforwardly evaluated as follows
\bes
&\avg{\ln \sum_{\boldsymbol{\sigma}}\exp\Big(\frac{\beta \alpha^2}{1+\alpha}\sum_{a,b} q_{ab} p_{ab} \sigma^a \sigma^b+\frac{2\beta}{1+\alpha} \xi \sum_a m^3_a \sigma^a\Big)}_\xi=\\
&=n\frac{\beta \alpha^2}{1+\alpha}(p_D-p q) + n \ln 2 + n \int Dx \:\ln \cosh \Big( \sqrt{2\frac{\beta \alpha^2}{1+\alpha} p q}\:x+ \frac{2\beta}{1+\alpha} m^3 \Big).
\ees
Putting all these results in the expression for the statistical pressure, we get the result \eqref{A}.

\section{Zero-temperature critical capacity analysis}
In order to estimate the zero-temperature critical capacity $\alpha_c(T=0)$, we start from the self-consistency equations \eqref{eq:selfc}. Upon eliminating the conjugate parameter $p$, we get
\begin{equation}\label{eq:self1}
\begin{split}
m  &=  \int Dx\,\tanh\left(\frac\beta{1+\alpha}\Big(\frac{\sqrt{2{ \alpha^3 q^3} }}{1-\frac{\beta \alpha}{1+\alpha}(1-q^2)}\:x + 2\:  m^3\Big)\right),\\
q  &=  \int Dx\,\tanh^2\left(\frac\beta{1+\alpha}\Big(\frac{\sqrt{2{ \alpha^3 q^3} }}{1-\frac{\beta \alpha}{1+\alpha}(1-q^2)}\:x + 2\:  m^3\Big)\right).
\end{split}
\end{equation}
In the limit $\beta\to \infty$, it is easy to check that $q\to1$, then $1-q^2 \to 0$ in the zero temperature limit. The quantity $C=\beta(1-q^2)$, which satisfies the self-consistency equation
\begin{equation}
C= \beta -\beta \Big(1-\frac{1+\alpha}{2\beta}\frac{\partial}{\partial (m^3)}\int Dx \tanh (g(m,q))\Big)^2,
\end{equation}
where $g(m,q)$ is the argument of the hyperbolic tangent in \eqref{eq:self1}, is finite in the $\beta\to\infty$ limit. Using $\tanh(\beta x)\to \text{sgn}(x)$ in the large $\beta$ limit, then the self-consistency equations can be evaluated as
\begin{equation}
\begin{split}
m&=\text{erf}\left(\frac{m^3}{\alpha^{3/2}}\Big(1-\frac\alpha{1+\alpha}(1-C)\Big)\right),\\
C&=(\alpha+1)\frac{1-\frac\alpha{1+\alpha}C}{\alpha^{3/2}}\frac2{\sqrt \pi}\exp\Big(-\frac{m^6}{\alpha^3}(1-\frac{\alpha}{1+\alpha}C)^2\Big).
\end{split}
\end{equation}
By introducing the quantity
\begin{equation}
t=\frac{m^3}{\alpha^{3/2}}\Big(1-\frac\alpha{1+\alpha}(1-C)\Big),
\end{equation}
after some rearrangements, we end up with a single equation
\begin{equation}\label{eq:zeroT}
t= \frac{1}{\alpha^{3/2}}\text{erf}^3(t)-\frac{2t}{\sqrt{\alpha\pi}}\exp(-t^2).
\end{equation}
Then, the critical storage capacity is the value of $\alpha$ leading to non-trivial solutions for Eq. \eqref{eq:zeroT}. A comparison between the two sides of the equation for various $\alpha$ values is reported in Fig. \ref{fig:zerot}.
\begin{figure}[h!]
	\centering
	\includegraphics[width=0.7\textwidth]{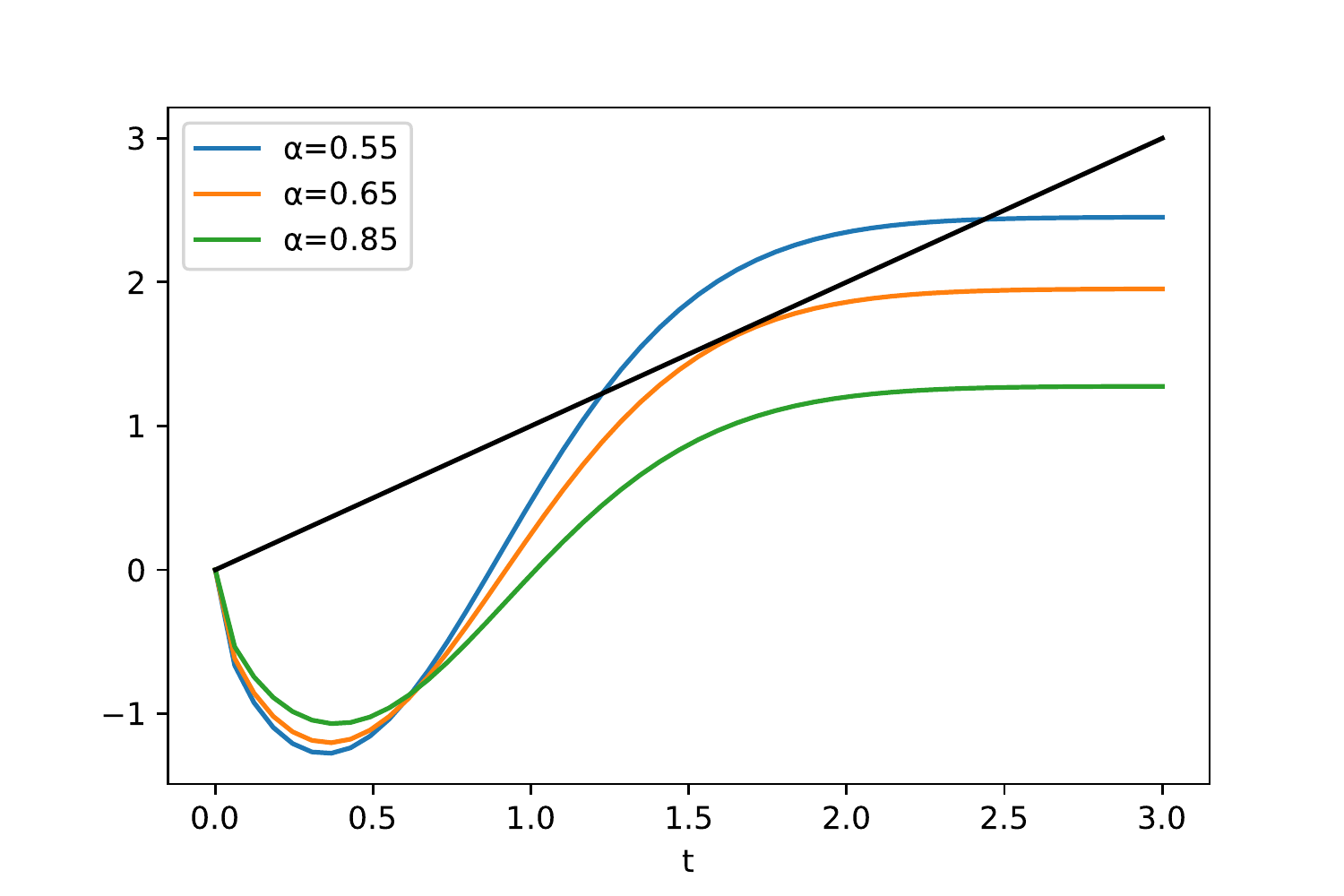}
	\caption{Comparison between the l.h.s. (black dashed line) and r.h.s. (blue solid curves) of Eq. \eqref{eq:zeroT} for $\alpha=0.55,0.65,0.85$.}\label{fig:zerot}
\end{figure}\\
By numerically solving the equation \eqref{eq:zeroT}, we found a critical storage capacity $\alpha_c(T=0)\simeq 0.651$, which is in perfect agreement with the phase diagram.

\section{Signal-to-noise analysis}

We here perform a signal-to-noise analysis, motivating the decomposition eq. \ref{decomposition}. Introducing the "internal" field $h_i$ seen by the $i$-th spin $\sigma_i$, defined as
	\begin{align}
	h_i=\frac{1}{2N^3} \sum_{\mu=1}^{K}\sum_{j,k,l=1}^{N} \eta_{ij}^\mu \eta_{kl}^\mu \sigma_j \sigma_k \sigma_l,
	\end{align}
	we can write the hamiltonian of the model as:
	\begin{align}
	H=-\sum_{i=1}^{N} h_{i} \sigma_i
	\end{align}
By virtue of the pattern decomposition, this field can be rewritten as
\begin{align}
h_i=\frac{1}{2N^3} \sum_{\mu=1}^{K}\sum_{j,k,l=1}^{N} \left(\xi_i^\mu \xi_j^\mu \xi_k^\mu \xi_l^\mu + \sqrt{K} \xi_i^\mu\xi_j^\mu J_{kl}^\mu + \sqrt{K} \xi_k^\mu\xi_l^\mu J_{ij}^\mu + K J_{ij}^\mu J_{kl}^\mu \right) \sigma_j \sigma_k \sigma_l
\end{align}
Probing the alignment to the pattern $\boldsymbol{\xi}^1=(\xi_1^1,..,\xi_N^1)$, we set $\boldsymbol{\sigma}=\boldsymbol{\xi}^1$, by which the following standard decomposition holds (we simply separate the "signal" characterized by $\mu=1$ from the "noise" $\mu>1$):
\begin{align}
h_i \sigma_i=\mathcal{S}+\mathcal{N}
\end{align}
where
	\begin{align}
\mathcal{S}=\frac{1}{2}\left[1+\frac{\sqrt{K}}{N}\sum_j J_{ij}^1 \xi_i^1\xi_j^1 + \frac{\sqrt{K}}{N^2} \sum_{k,l} J_{kl}^1 \xi_k^1\xi_l^1+ \frac{K}{N^3} \sum_{j,k,l}  \xi_i^1\xi_j^1\xi_k^1\xi_l^1 J_{il}^1 J_{kl}^1\right]
\end{align}
is the "signal" and 
\begin{equation}\label{eq:6}
\mathcal{N}=\frac{1}{2N^3}\sum_{\mu>1}^{K}\sum_{j,k,l=1}^{N} \left(\xi_i^\mu\xi_i^1 \xi_j^\mu \xi_j^1 \xi_k^\mu \xi_k^1 \xi_l^\mu \xi_l^1+ \sqrt{K} \xi_i^\mu \xi_i^1 \xi_j^\mu \xi_j^1 \xi_k^1 \xi_l^1 J_{kl}^\mu + \sqrt{K} \xi_k^\mu \xi_k^1 \xi_l^\mu \xi_l^1 \xi_i^1 \xi_j^1 J_{ij}^\mu + K \xi_i^1 \xi_j^1 \xi_k^1 \xi_l^1 J_{ij}^\mu J_{kl}^\mu \right)  
\end{equation}
the "noise". Recall that the $J_{ij}^\mu$ tensors are all $i.i.d.$ variables distributed as $\mathcal{N}(0,1)$. In order to compute the signal-to-noise ratio $\mathcal{S}/\mathcal{N}$, we firstly perform the standard Gaussian expectation $\mathbb{E}$ over the signal $\mathcal{S}$. This results in 
\begin{align}
\mathbb{E}\psq{\mathcal{S}}=\frac{1}{2}\left(1+\frac{K}{N^2}\right)\to \frac{1}{2}
\end{align}
as $N\to\infty$ in the thermodynamic limit. In order to get this result we consider that 
\begin{align}
\mathbb{E}\psq{J_{ij}^1}=\mathbb{E}\psq{J_{kl}^1}=0
\end{align}
and 
\begin{align}
\mathbb{E}\psq{J_{ij}^1 J_{kl}^1}=\delta_{ik} \delta_{jl}
\end{align}
and the fact that the  products similar to $J_{ij}^1 \xi_i^1\xi_j^1$ give new $i.i.d$ $\mathcal{N}(0,1)$ variables as the $J$'s are.\\
Consider now the noise $\mathcal{N}$. The first term in the parenthesis (\ref{eq:6}) can be decomposed in a sum of various contributions, given the four summations in $\mu,j,k,l$. We get a contribution by setting $j=k=l=i$, which is of order $\mathcal{O}(N^{-2})$ given the overall factor $1/N^3$ in front of each term in the noise; we have several contributions from $l\neq k$ with $k=j=i$, and cyclic permutations (i.e. $l\neq j$ with $k=j=i$ and so on), which overall give a contribution of order $\mathcal{O}(N^{-2})$; then we have to consider the terms coming from $l\neq i, k\neq i, j=i$ and similar, giving $\mathcal{O}(N^{-3/2})$ contributions and, the remaining ones coming from $l \neq i, k\neq i, j\neq i$ and similar, which are $\mathcal{O}(N^{-1})$. We see therefore that, in the thermodynamic limit, the first term in the noise is zero.\\
The remaining terms have to be evaluated via the Gaussian expectation operator, therefore we can easily apply similar considerations to those used in the evaluation of $\mathbb{E}\psq{\mathcal{S}}$. This results in vanishing contributions from the second and the third term in the noise decomposition in the thermodynamic limit. The only non-zero contribution comes from the last term, the fourth, which however is non-zero only for $k=i, j=l$, giving $\alpha^2/2$.\\
The ratio $\mathcal{S}/\mathcal{N}$ can now easily computed, giving $1/\alpha^2$, which is of order $\mathcal{O}(1)$. It can be shown that this is the minimal $\mathcal{S}/\mathcal{N}$ value given the decomposition in eq. \ref{decomposition}: attempting to overtake the linear load $K=\alpha N$ (e.g. by considering super-linear regimes such as $K\sim N^2$) leads to a vanishing signal to noise ratio in the thermodynamic limit. Hence, in super-linear regimes the network cannot retrieve any pattern of information: our decomposition eq. \ref{decomposition} is therefore the worst scenario from the network's point of view, i.e. it gives the maximal noise to the network maintaining its retrieval capability.

\end{document}